\documentclass[letterpaper, 12pt]{article}

\usepackage[utf8]{inputenc}
\usepackage{geometry}
\usepackage[numbers, sort]{natbib}
\usepackage{graphicx}
\usepackage{amsmath,amsthm,amsfonts,amssymb}
\usepackage{enumitem}
\usepackage{tikz}
\usepackage{mdframed}
\usepackage{algorithmic}
\usepackage[ruled,vlined]{algorithm2e}
\usepackage{dsfont}
\usepackage{xcolor}
\usepackage[hidelinks]{hyperref}

\newtheorem{definition}{Definition}
\newtheorem{theorem}{Theorem}
\newtheorem{assumption}{Assumption}
\newtheorem{claim}{Claim}
\newtheorem{lemma}{Lemma}
\newtheorem{corollary}{Corollary}

\providecommand{\mnorm}[1]{\ensuremath{\left\lvert#1\right\rvert}}

\def\H{\mathcal{H}}

\def\G{\mathcal{G}}

\def\V{\mathcal{V}}

\begin{document}

\title{Byzantine fault-tolerant distributed set intersection with redundancy }
\author{Shuo Liu \hspace{0.5in} Nitin H. Vaidya\\\small{Georgetown University}\\\small{\texttt{\{sl1539, nitin.vaidya\}@georgetown.edu}}}
\date{}

\maketitle

\begin{abstract}
    In this report, we study the problem of Byzantine fault-tolerant distributed set intersection and the importance of redundancy in solving this problem. Specifically, consider a distributed system with $n$ agents, each of which has a local set. There are up to $f$ agents that are Byzantine faulty. The goal is to find the intersection of the sets of the non-faulty agents. 
    
    We derive the Byzantine set intersection problem from the Byzantine optimization problem. We present the definition of $2f$-redundancy, and identify the necessary and sufficient condition if the Byzantine set intersection problem can be solved if a certain redundancy property is satisfied, and then present an equivalent condition. We further extend our results to arbitrary communication graphs in a decentralized setting. Finally, we present solvability results for the Byzantine optimization problem, inspired by our findings on Byzantine set intersection. The results we provide are for synchronous and asynchronous systems both.
\end{abstract}

\section{Introduction}
\label{sec:intro}

Consider a distributed system of $n$ agents, where each agent $i$ has a local set $X_i$. The goal of the Byzantine fault-tolerant distributed set intersection problem for the non-faulty agents is to collectively compute the intersection of the $X_i$'s for all non-faulty agents $i$'s despite the presence of up to $f$ faulty agents, i.e., to find
\begin{align}
    \bigcap_{i\in\V\backslash F}X_i
    \label{eqn:goal-set-intersection}
\end{align}
where $\V$ is the set of all agents, and $F$ is the set of Byzantine faulty agents with $\mnorm{F}\leq f$. A Byzantine faulty agent may behave arbitrarily, differing from the prescribed algorithm. The Byzantine faulty agents may also work in a collaborative effort to disrupt the algorithm. 
We call an algorithm that achieves this goal $f$-resilient, or formally,
\begin{definition}[$f$-resilient]
    \label{def:f-resilient-set-intersection}
    We say an algorithm for Byzantine set intersection is \emph{$f$-resilient} if it outputs the intersection of the sets of all non-faulty agents, defined in \eqref{eqn:goal-set-intersection}, in the presence of up to $f$ Byzantine agents.
\end{definition}

There are two commonly seen models for the architecture of a distributed system. A \textit{centralized} system consists of a (trusted) server and $n$ agents. Communications happen only between the server and each agent, but not between two agents. On the other hand, a \textit{decentralized} system consists of only $n$ agents, and communications happen between pairs of agents directly. In this case, a directed graph $G(\V,\mathcal{E})$ can be used to describe the communication graph, where $\V$ is the set of agents, and $\mathcal{E}$ is the set of directed edges where $(i,j)\in\mathcal{E}$ if agent $i$ can send messages to agent $j$. Both structures will be discussed later in this report. We also note that throughout this paper we only discuss synchronous algorithms, unless specifically stated otherwise.

As we will elaborate in detail in Section~\ref{sec:optimization}, the problem of Byzantine distributed set intersection is related to the problem of Byzantine distributed optimization. Specifically, prior work has shown that the necessary condition for solving a Byzantine optimization problem exactly is to satisfy \textit{$2f$-redundancy} (defined later) \cite{gupta2020fault}, which indicates that the minimum point sets of non-faulty agents intersect, and the minimum point set of the aggregate cost functions of non-faulty agents is the same as the intersection of the minimum point set of these agents.

In this report, we first demonstrate the relationship between Byzantine set intersection and Byzantine optimization problems in Section~\ref{sec:optimization}. In Section~\ref{sec:set-inter}, we then propose a similar redundancy property and present the condition on the graph under which an $f$-resilient set intersection algorithm with certain constraints exists when the redundancy property is satisfied. We will also show an equivalent condition in Section~\ref{sec:equiv}. We extend our findings to unconstrained algorithms and asynchronous cases in the remainder of Section~\ref{sec:set-inter}. In Section~\ref{sec:generalized-graph}, we further extend our results to generalized communication graphs. Lastly, we present some necessity and sufficiency results on Byzantine optimization in Section~\ref{sec:to-optimization} inspired by results in the previous sections.
\section{Related work}
\label{eqn:related}


There are previous studies on problems related to the Byzantine distributed set intersection problems. \textit{Certified propagation algorithms} \citep{koo2004broadcast, pelc2005broadcasting, tseng2015broadcast, tseng2019reliable} are proposed for a problem in the presence of $f$-local Byzantine faults, where there is a source agent in a distributed network with an initial input, and that value needs to be transmitted to all other agents, knowing that for every non-faulty agent, there are up to $f$ incoming neighbors are faulty. The set intersection problem can be adapted to a propagation problem: imagine there is a ``virtual'' source agent in a distributed network for a set intersection problem, where for each value $y\notin\bigcap_{i\in\V}X_i$, it needs to propagate a value 0 to all other agents, in other words, all agents should end up with knowing that $y$ is not in $\bigcap_{i\in\V}X_i$. However, the formulation in \citep{tseng2015broadcast} is different from our setting, since in \citep{tseng2015broadcast} the neighbors of the source agent are fixed, while in an adaptation of a set intersection problem, the neighbor of the virtual source agent can be different for different $y\notin\bigcap_{i\in\V}X_i$.

Another line of work studies the \textit{lattice agreement} problem \cite{attiya1995atomic, zheng2018lattice}. On a complete graph connecting $n$ agents, each agent proposes a local value in a \textit{finite join semi-lattice}, or simply lattice. The goal is for each agent to decide on a value in the lattice such that the values of every pair of agents are comparable. The sets in a set intersection problem can be defined as a lattice with subset relation being the partial order and the union operation of sets being the join operation of the lattice. In this way, the problems of lattice agreement and set intersection are similar. However, lattice agreement only requires each agent to decide on \textit{comparable} values, while the set intersection problem requires agreement on the \textit{same} intersection. Furthermore, in general, the set intersection problem we study in this paper does not require a complete communication graph. Also, we explore a special case where the input sets have a certain form of redundancy, introduced in later sections.

In \citep{vaidya2012iterative}, the authors studied a problem of iterative Byzantine consensus in arbitrary directed graphs. Our work is similar in terms of the communication graph we consider. We also use similar analysis techniques in Sections~\ref{sec:equiv} and~\ref{sec:generalized-graph} called \textit{reduced graph} and \textit{source component}. 
\citet{su2016asynchronous} studied the problem of asynchronous distributed hypothesis testing with crash failures. This paper also uses the reduced graph and source component technique. However, in both papers, only one source component in each reduced graph is allowed, while our analysis allows multiple source components.

\citet[Chapter 4.3]{su2020defending} studied Byzantine consensus problem in $m$-dimension with the reduced graph technique, as well as Byzantine non-Bayesian Learning in \citep[Chapter 4.4]{su2020defending}. In both discussions, the creation of a reduced graph is related to the dimension $m$ of the target vector -- each agent may remove up to $mf$ incoming links. This is a much stronger requirement compared to our analysis on set intersection and optimization, where we only remove up to $f$ incoming links to create a reduced graph.
\citet{vyavahare2019distributed} studied distributed learning with adversarial agents using the reduced graph technique. Similar to our analysis, multiple source clusters are allowed in each reduced graph. However, similar to in \cite{su2020defending}, reduced graphs are allowed to remove up to $mf$ incoming links at each node, where $m$ is the dimension.

\citet{mitra2020new} studies distributed hypothesis testing and non-Bayesian learning with Byzantine agents under the decentralized architecture. The set intersection problem is closely related to hypothesis testing, where the fixed true state of the world is the intersection and the observations of the agents are similar to the local sets of each agent. \citep{mitra2020new} provided necessary and sufficient conditions on the structure of the communication graph for a Byzantine resilient algorithm to exist. However, the conditions are defined for a certain joint observation profile, while in our results the conditions are for all possible input sets. Also, \citep{mitra2020new} studied the $f$-local fault model, i.e., there are up to $f$ Byzantine agents among the incoming neighbors of each agent, while we study $f$-global model in this report, where there are up to $f$ Byzantine agents in the system.

\citet{gupta2021byzantine} studied Byzantine optimization problems in a decentralized architecture and examined the importance of $2f$-redundancy in solving Byzantine optimization problems specifically. We discuss $2f$-redundancy and optimization problem on a decentralized architecture in Section~\ref{sec:to-optimization} as well, but \cite{gupta2021byzantine} only examines a fully connected communication graph (or a ``peer-to-peer network''), while our results apply to more generalized communication graphs.
\section{From optimization to set intersection}
\label{sec:optimization}

The problem of Byzantine distributed set intersection is related to the problem of Byzantine distributed optimization. Recall that we defined $\V$ to be the set of all agents, and $F$ to be the set of Byzantine agents in a given execution with $\mnorm{F}\leq f$. In Byzantine optimization \cite{gupta2020fault, yin2018byzantine, chen2017distributed, blanchard2017machine}, each agent $i$ has a local cost function $Q_i(x)$, and the problem asks the non-faulty agents to collaboratively compute an output $\widehat{x}$, such that 
\begin{align}
    \widehat{x}\in \arg\min_x\sum_{i\in\V\backslash F}Q_i(x),
\end{align}
despite the presence of up to $f$ Byzantine faulty agents. Algorithms that achieve this goal are called $f$-resilient \citep{gupta2020fault}, or formally
\begin{definition}[$f$-resilience \cite{gupta2020fault}]
    \label{def:f-resilience}
    A distributed optimization algorithm is $f$-resilient, if it outputs a minimum point of the aggregate cost function of all non-faulty agents, despite the presence of up to $f$ Byzantine faulty agents.
\end{definition}
\citet{gupta2020fault} have shown that, assuming the cost functions are convex and differentiable, and a minimum point in $\arg\min_x\sum_{i\in\V\backslash F}Q_i(x)$ always exists, it is necessary to have $2f$-redundancy for an $f$-resilient algorithm to exist, where $2f$-redundancy is defined as follows:
\begin{definition}[$2f$-redundancy \cite{gupta2020fault}]
    \label{def:redundancy-exact-original}
    For a given set of non-faulty agents $\H$, their local cost functions are said to satisfy the $2f$-redundancy if and only if for every subset $S\subseteq\H$ with $\mnorm{S}\geq n-2f$, 
    \begin{equation}
         \arg\min_x\sum_{i\in S}Q_i(x) = \arg\min_x\sum_{i\in\H}Q_i(x).
         \label{eqn:2f-redundancy-original}
    \end{equation}
\end{definition}
\noindent If a group of $n$ agents always satisfy $2f$-redundancy in Definition~\ref{def:redundancy-exact-original} no matter which up to $f$ agents   are faulty, by \eqref{eqn:2f-redundancy-original}, we must also have
\begin{equation}
     \arg\min_x\sum_{i\in S}Q_i(x) = \arg\min_x\sum_{i\in\V\backslash F}Q_i(x),
\end{equation}
where $F\subset\V$ with $\mnorm{F}\leq f$ is an arbitrary subset. It is also shown in \cite[Section 3.1]{gupta2020fault} that it is necessary to have $\bigcap_{i\in\V}\arg\min_xQ_i(x)\neq\varnothing$ for an $f$-resilient optimization algorithm to exist. Therefore, we can formally define the following:
\begin{definition}[$2f$-redundancy]
    \label{def:redundancy-exact}
    For a set $\V$ of $n$ agents, their local cost functions are said to satisfy the $2f$-redundancy if and only if (i) $\bigcap_{i\in\V}\arg\min_xQ_i(x)\neq\varnothing$, and (ii) for every subset $S\subseteq\V$ with $\mnorm{S}\geq n-2f$, 
    \begin{equation}
         \arg\min_x\sum_{i\in S}Q_i(x) = \arg\min_x\sum_{i\in\V}Q_i(x).
         \label{eqn:2f-redundancy}
    \end{equation}
\end{definition}
We call the property defined above \textbf{Property A}. By the discussion above, if the total number of agents is $n$ and the number of faulty agents is no more than $f$, part (ii) of Property A is implied by the property in Definition~\ref{def:redundancy-exact-original}, and therefore also a necessary condition for $f$-resilient optimization, while the necessity of part (i) is derived in \citep{gupta2020fault}. Thus, Property A \textit{is} necessary to have $f$-resilient optimization algorithms as well, for convex and differentiable cost functions.

We also have the following lemma that will be used later:
\begin{lemma}[\cite{gupta2020fault}]
    \label{lemma:minimum-set-intersection}
    For any non-empty subset $S$ of $\{1,...,n\}$, if 
    \begin{align}
        \bigcap_{i\in S}\arg\min{Q_i(x)} \neq \varnothing,
    \end{align}
    then 
    \begin{align}
        \arg\min\sum_{i\in S}Q_i(x) = \bigcap_{i\in S}\arg\min{Q_i(x)}.
    \end{align}
\end{lemma}
Let us denote, for any agent $i$ and any set of agents $S$,
\begin{align}
    X_i &\triangleq \arg\min_x Q_i(x), \\
    X_S &\triangleq \arg\min_x\sum_{i\in S}Q_i(x).
\end{align}

By Lemma~\ref{lemma:minimum-set-intersection}, with Property A, finding a point in the set $\arg\min_x\sum_{i\in\V\backslash F}Q_i(x)$ is equivalent to finding a point in the intersection of the minimum set of the cost functions of the non-faulty agents $\bigcap_{i\in\V\backslash F}\arg\min_xQ_i(x)$.

Lemma~\ref{lemma:minimum-set-intersection} and Property A also imply that if the cost functions are convex and differentiable, for all $S\subseteq\V$ we have
\begin{align}
    X_S&=\bigcap_{i\in S}X_i.
\end{align}
Therefore, Property A of $2f$-redundancy in Definition~\ref{def:redundancy-exact} is equivalent to the following:
\begin{definition}[$2f$-redundancy, \textbf{Property B}]
    \label{def:set-redundancy-1}
    The cost functions of a set $\V$ of $n$ agents are said to satisfy $2f$-redundancy, if $\bigcap_{i\in 
    \V}X_i\neq\varnothing$, and for any subsets $S\subseteq\V$ with $\mnorm{S}\geq n-2f$, $\bigcap_{i\in S}X_i = \bigcap_{i\in\V}X_i$.
\end{definition}

Note that Property B is a property on the minimum set of each agent, while Property A is a property on the cost functions. We can also propose another redundancy definition in terms of the minimum set of each agent.
\begin{definition}[$2f$-redundancy, \textbf{Property C}]
\label{def:set-redundancy-2}
    The cost functions of a set $\V$ of $n$ agents are said to satisfy $2f$-redundancy, if $\bigcap_{i\in 
    \V}X_i\neq\varnothing$, and for any point $y\notin \bigcap_{i\in\V}X_i$, there exist at least $2f+1$ agents $i$, such that $y\notin X_i$.
\end{definition}
Lemmas~\ref{lemma:alternative-condition-necessity} and \ref{lemma:alternative-condition-sufficiency} show that Property C is also equivalent to Property B.

\begin{lemma}
    Suppose Property B is satisfied. We have (i) $\bigcap_{i\in\V}X_i\neq\varnothing$, and (ii) for any point $y\notin \bigcap_{i\in\V}X_i$, there exist at least $2f+1$ agents, such that for any agent $i$ of those $2f+1$ agents, $y\notin X_i$.
    \label{lemma:alternative-condition-necessity}
\end{lemma}

\begin{proof}
    Part (i) of the lemma is the same as in Property B.
    
    For Part (ii), we prove it by contradiction. Suppose that there exists $y_0\notin \bigcap_{i\in \V}X_i$, such that there are no more than $2f$ agents $i$'s with $y_0\notin X_i$. In other words, there are no less than $n-2f$ agents $j$'s, such that $y_0\in X_j$. Suppose these agents form a set $S$. Thus, $y_0\in\bigcap_{i\in S}X_i$. 
    With Property B, $\bigcap_{i\in S}X_i=\bigcap_{i\in\V}X_i$ for all $S$ with $\mnorm{S}\geq n-2f$. Therefore, $y_0\in\bigcap_{i\in S}X_i$ implies $y_0\in \bigcap_{i\in\V}X_i$, which contradicts to our assumption that $y_0\notin\bigcap_{i\in\V}X_i$. Hence the proof.
\end{proof}

\begin{lemma}
    Suppose the cost functions of a group $\V$ of $n$ agents satisfy the following: (i) $\bigcap_{i\in\V}X_i\neq\varnothing$, and (ii) for any point $y\notin\bigcap_{i\in\V}X_i$, there exist at least $2f+1$ agents, such that for any agent $i$ of those $2f+1$ agents, $y\notin X_i$. Then Property B is also satisfied.
    \label{lemma:alternative-condition-sufficiency}
\end{lemma}
\begin{proof}
    The proof is by contradiction. 
    Note that we have $\bigcap_{i\in\V}X_i\subseteq\bigcap_{i\in S}X_i$.
    
    Suppose Property B is not satisfied. That is, there exists a set $S\subset\V$ with $\mnorm{S}\geq n-2f$, $\bigcap_{i\in S}X_i\neq\bigcap_{i\in\V}X_i$. Since $\bigcap_{i\in\V}X_i\subseteq\bigcap_{i\in S}X_i$, it follows that there exists $x\in\bigcap_{i\in S}X_i$, such that $x\notin\bigcap_{i\in\V}X_i$. Therefore, $x\in X_i$, where $i\in S$. There are at least $n-2f$ agents in $S$, meaning there are at most $2f$ agents $j$ with $x\notin X_j$. This contradicts part (ii) of the assumption in the lemma. Hence, the proof.
\end{proof}

\section{Byzantine set intersection with $2f$-redundancy}
\label{sec:set-inter}

As discussed in the previous section, the $2f$-redundancy property for Byzantine optimization problems, i.e., Property A, can be transformed into properties for set intersection, i.e., Properties B and C. In this section, we will study the problem of Byzantine set intersection directly. Specifically, we explore the necessary condition for an $f$-resilient algorithm to exist if $2f$-redundancy is satisfied.

From a distributed set intersection perspective, we first restate the two redundancy properties we derived in the previous section. Recall the definition of our problem in Section~\ref{sec:intro}, where there are $n$ agents in the distributed system, and each agent $i$ has a local set $X_i$. The goal is to find the intersection of all non-faulty agents $X_{\V\backslash F} = \bigcap_{i\in\V\backslash F}X_i$ in the presence of up to $f$ Byzantine agents, where $\V$ is the set of all agents and $F$ is the set of all Byzantine faulty agents in any execution.

\begin{definition}[$2f$-redundancy, \textbf{Property B}]
    \label{def:set-redundancy-1-1}
    The local sets of a set $\V$ of $n$ agents are said to satisfy $2f$-redundancy, if $\bigcap_{i\in\V}X_i\neq\varnothing$, and for any set of agents $S$ with $\mnorm{S}\geq n-2f$, $\bigcap_{i\in\V}X_i = \bigcap_{i\in S}X_i$.
\end{definition}
\begin{definition}[$2f$-redundancy, \textbf{Property C}]
\label{def:set-redundancy-2-1}
    The local sets of a set $\V$ of $n$ agents are said to satisfy $2f$-redundancy, if $\bigcap_{i\in\V}X_i\neq\varnothing$, and for any point $y\notin \bigcap_{i\in\V}X_i$, there exist at least $2f+1$ agents $i$, such that $y\notin X_i$.
\end{definition}
Note that following the discussions in Section~\ref{sec:optimization}, the two properties above are equivalent.

\subsection{Centralized architecture}
In a centralized architecture, there is a trusted server that communicates with all $n$ agents, and agents do not communicate with each other directly. In this case, $2f$-redundancy alone is sufficient to achieve $f$-resilience. Specifically, consider the following algorithm:

\begin{enumerate}[label=\textbf{Step \arabic*}]
    \item The server requires each agent $i$ to send its local set $X_i$. A non-faulty agent sends the actual local set, while a Byzantine agent $j$ may send an arbitrary set. 

    The server receives a set $Y_i$ from each agent $i$. For every non-faulty agent $i$, $Y_i=X_i$. For every Byzantine agent $j$, $Y_j$ can be an arbitrary set.
    \item For each set $T$ of agents, $\mnorm{T}=n-f$, calculate the following:
        \begin{itemize}
            \item The set $Y_T = \bigcap_{i\in T}Y_i$, 
            \item The set $Y_{\widehat{T}}=\bigcap_{i\in \widehat{T}}Y_i$ for all $\widehat{T}\subset T$, $\mnorm{\widehat{T}}\geq n-2f$.
        \end{itemize}
    \item If there exists a set $T$ of size $n-f$, such that for all $\widehat{T}\subset T$, $\mnorm{\widehat{T}}\geq n-2f$, $Y_{\widehat{T}}$ are equal, output the set $Y_T$.
\end{enumerate}

Now we are going to prove that if Property B of $2f$-redundancy stated in Definition \ref{def:set-redundancy-1-1} is satisfied, the algorithm is $f$-resilient, per Definition~\ref{def:f-resilient-set-intersection}.
\begin{proof}
    Suppose in any execution, the set of all non-faulty agents is $S$. Since there are up to $f$ faulty agents, $\mnorm{S}\geq n-f$.
    Given Property B, for $S$ and any subset $\widehat{S}\subset S$ with $\mnorm{S}\geq n-f$ and $\mnorm{\widehat{S}}\geq n-2f$, 
    \begin{align}
        \bigcap_{i\in\V}X_i = \bigcap_{i\in S}X_i = \bigcap_{i\in\widehat{S}}X_i.
        \label{eqn:centralized-equal-sets}
    \end{align}
    Consider the set $T$ of agents that defines the output of the algorithm $Y_T$. Let us denote $\widehat{S}' = S\cap T$. Since $\mnorm{T}=n-f$, we have $\mnorm{\widehat{S}'}\geq n-2f$. Also, $\widehat{S}'\subset S$. Therefore, by \eqref{eqn:centralized-equal-sets}, $\bigcap_{i\in\V}X_i = \bigcap_{i\in\widehat{S}'}X_i$. By the definition of the algorithm, since $\widehat{S}'\subset T$, $Y_T = \bigcap_{i\in\widehat{S}'}X_i$. Therefore, the output of the algorithm $Y_T = \bigcap_{i\in\V}X_i$. That is, the algorithm is correct if Property B is satisfied, in the presence of up to $f$ faulty agents.
\end{proof}

\subsection{Decentralized architecture with constrained algorithms}
\label{sub:constrained}

In a decentralized architecture, there is no server and the $n$ agents communicate directly to each other. The \textit{communication graph} $G=(\V,\mathcal{E})$ can be used to describe the connections of the agents, where $\V$ is the set of $n$ agents, and $\mathcal{E}$ is a set of directed edges. An edge $(i,j)\in\mathcal{E}$ represents that agent $i$ can send information to agent $j$. Let us also define $N_i^-=\left\{j|(i,j)\in\mathcal{E}\right\}$ for any agent $i$ to be the set of \textit{incoming neighbors} of $i$. Note that throughout the remainder of this report, we define $i\notin N_i^-$ for all $i\in\V$ (or equivalently $(i,i)\notin\mathcal{E}$), but agent $i$ knows its local set nonetheless. 

In an iterative algorithm starting with iteration 0, let us denote $X_i^t$ to be the local set of agent $i$ at the beginning of iteration $t$ with $X_i^0 \triangleq X_i$. Also, recall that we denote by $F$ the set of faulty agents in any execution. 

We consider in this section a class of constrained algorithms. In particular, a \textit{correct} $f$-resilient set intersection algorithm must satisfy the following criteria:
\begin{itemize}
    \item \textbf{Validity}: $\forall t>0$, $X_i^t\subseteq X_i^{t-1}$.
    \item \textbf{Convergence}: There exists a finite $\tau$ such that after $\tau$ iterations, $X_i^\tau=X_{\V\backslash F}$, for all $i\in\V\backslash F$.
\end{itemize}
Note that the algorithms we discuss here are constrained to maintain only a local set $X_i^t$, and no additional state. In each iteration, each agent may send their own set to their outgoing neighbors, receive their incoming neighbors' local sets, respectively, and then update its local set according to the information received.

Now we present a necessary condition on the underlying communication graph $G$ for the system to be able to solve the Byzantine set intersection problem. The graph should be sufficient to solve the problem for all inputs that satisfy the $2f$-redundancy condition. In particular, recall Property C in Definition \ref{def:set-redundancy-2-1}, where we assume that for any point $y\notin \bigcap_{i\in\V}X_i$, there are at least $2f+1$ agents $i$, such that $y\notin X_i$.

\begin{theorem}[Necessity]
    \label{thm:necessity-set-intersection}
    Assume that the number of agents $n\geq 2f+2$. Suppose $2f$-redundancy is satisfied. A constrained set intersection algorithm exists only if for any partition of $G$ into $L, R, F$, $\mnorm{F}\leq f$, the following conditions hold:
    \begin{enumerate}[label=\alph*)]
        \item If $\mnorm{{R}}\geq f+1$, there exists an agent $i\in L$, $\mnorm{N_i^-\cap{R}}\geq f+1$, and
        \item If $\mnorm{L}\geq f+1$, there exists an agent $i\in {R}$, $\mnorm{N_i^-\cap L}\geq f+1$.
    \end{enumerate}
\end{theorem}
\begin{proof}
    We first prove condition a) is necessary. The proof is by contradiction. Suppose that a correct set intersection algorithm on graph $G$ exists, and there exists a partition $L, {R}, F$ of $G$ with $\mnorm{F}\leq f$ and $\mnorm{{R}}\geq f+1$, such that for each $i$ in $L$, $\mnorm{N_i^-\cap{R}}\leq f$. 

    Consider the following transformation of the three sets $L$, $R$, and $F$: identify $L_1\subset L$ and $R_1\subset R$, such that $L\backslash L_1\neq\varnothing$, $\mnorm{R\backslash R_1}\geq f+1$, and $\mnorm{F\cup L_1\cup R_1}=f$, and let $L'=L\backslash L_1$, $R'=R\backslash R_1$, and $F'=F\cup L_1\cup R_1$. $L'$, $R'$, $F'$ is also a partition of $G$, with $\varnothing\neq L'\subseteq L$, $R'\subseteq R$ and $F\subseteq F'$, $|R'|\geq f+1$, $|F|=f$, and for each $i$ in $L'$, $\mnorm{N_i^-\cap{R'}}\leq f$. Note that this transformation is always feasible because $n\geq 2f+2$.
    
    

    Consider input sets such that there exists some point $y$ for which there are at least $2f+1$ agents $i$ such that $y\not\in X_i$. 
    Specifically, for each agent $i\in F'\cup R'$, $y\notin X_i$. Since $|F'|=f$, $|R'|\geq f+1$, there are at least $2f+1$ agents whose local sets do not have $y$.
    Also, for each agent $i\in L'$, $y\in X_i$.
    
    For any agent $l\in L'$, there are at most $f$ agents in $R'$ that can send information to $l$. Suppose that in an execution, all agents in $F'$ are faulty, and any faulty agent in $F'$ will always send to $l$ a set that contains $y$.
    Any agent $l\in L'$ cannot distinguish the following two scenarios:
    \begin{enumerate}[nosep, label=\roman*)]
        \item The actual situation, where agents in $F'$ are faulty, agents in $R'$ are non-faulty, and $y\notin X_i$ for $i\in R'$. For each agent $i\in L'$, $y\in X_i$. 
        In this case, $y$ is not in $\bigcap_{i\in\V\backslash F'}X_i$, the intersection of all non-faulty local sets, and should be removed from $X_l$.
        \item Another possible situation from $l$'s perspective, where the agents in $N_l^- \cap R'$ are faulty, the remaining agents are non-faulty, and for each non-faulty agent $i$, $y\in X_i$. In this case, $y$ is in $\bigcap_{i\in\V\backslash(N_l^-\cap R')}X_i$, which is the intersection of all non-faulty local sets.
    \end{enumerate}
    Note that both scenarios are compatible with the $2f$-redundancy property: in Scenario i), there are at least $2f+1$ agents whose local sets do not have $y$, while in Scenario ii), all the non-faulty agents have $y$ in their local sets.
    
    Let us assume that the validity condition is satisfied, that is, $\forall t$, $X_l^t\subseteq X_l^{t-1}$.
    If in any iteration $t$, agent $l$ chooses to remove $y$ from $X_l$, then its output would not satisfy the convergence condition in the second scenario above. On the other hand, if agent $l$ does not remove $y$ from $X_l$ in any iteration, then its output would not satisfy the convergence condition in the first scenario. Thus, agent $l$'s output cannot be guaranteed to be correct in all executions.
    This is a contradiction with the assumption that a correct algorithm exists. Thus condition a) is necessary.
    
    Similarly, it can be shown that condition b) is also necessary.
\end{proof}

The necessary condition in Theorem~\ref{thm:necessity-set-intersection} is also sufficient for an $f$-resilient algorithm to exist if $2f$-redundancy is satisfied.
\begin{theorem}[Sufficiency]
    \label{thm:sufficiency-set-intersection}
    Suppose $2f$-redundancy is satisfied and $n\geq 2f+2$, a constrained set intersection algorithm exists if for any partition of $\G$ into $L, R, F$, $\mnorm{F}\leq f$, 
    \begin{enumerate}[label=\alph*)]
        \item If $\mnorm{{R}}\geq f+1$, there must exist an agent $i\in L$, $\mnorm{N_i^-\cap{R}}\geq f+1$, and
        \item If $\mnorm{L}\geq f+1$, there must exist an agent $i\in {R}$, $\mnorm{N_i^-\cap L}\geq f+1$.
    \end{enumerate}
\end{theorem}

\begin{proof}
    The proof is constructive. Consider the following algorithm:

    Each agent $i$ sets its set $X_i^0$ to be its input set. In each iteration $t$,
    \begin{itemize}
        \item[\textbf{Step 1}] Each agent $i$ sends the set $X_i^t$ to all agents $j$ with $(i,j)\in\mathcal{E}$. Byzantine faulty agent may send arbitrary set.
        \item[\textbf{Step 2}] Each agent $j$ receives sets from all agents in $N_j^-$. For each agent $k\in N_j^-$, denote the set $j$ receives by $Z_k^t$. For a non-faulty agent $k$, $Z_k^t = X_k^t$. For a Byzantine agent $l$, $Z_l^t$ can be an arbitrary set.
        \item[\textbf{Step 3}] Each agent $i$ updates its set by the following rule: Initiate the set $Y_i^t$ to be empty. For any point $y\in X_i^t$, if there are more than $f$ agents $j\in N_i^-$ such that $y\notin Z_j^t$, add $y$ to the set $Y_i^t$. Then, the local set of $i$ is updated as
        \begin{align}
            X_i^{t+1} \leftarrow X_i^t\backslash Y_i^t.
        \end{align}
    \end{itemize}

    The algorithm by design satisfies the validity criteria.

    Suppose $F\subset\V$ is the set of all faulty agents. Now we show that in any execution, for any point $y\notin \bigcap_{i\in\V\backslash F}X_i$, $y$ will eventually be removed from the sets of all non-faulty agents. 
    
    Consider the partition $L_0, R_0, F$ of $G$, where $L_0$ is the set of all non-faulty agents that do not have $y$ in their sets at the beginning of Iteration 0. Since $\mnorm{F}\leq f$, by $2f$-redundancy, $\mnorm{L_0}\geq f+1$. By condition b) of this theorem, there exists an agent $r\in R_0$, such that there are at least $f+1$ incoming neighbors in $L_0$, and these neighbors will send their sets to $r$ in iteration 0, which do not include the point $y$. Therefore, according to the algorithm, $r$ receives at least $f+1$ sets that do not contain $y$, and it removes $y$ from its set in Iteration 0.

    In the next iteration, consider the partition $L_1 = L_0 \cup \{r\}$, $R_1 = L_0 \backslash \{r\}$, $F$ of $G$. Similarly, there will be another agent in $R_1$ whose set will remove $y$ in Iteration 1. Since $\mnorm{R_0}< n-2f$, in at most $n-2f$ iterations, $y\notin\bigcap_{i\in\V\backslash F}X_i$ will be removed from the sets of all non-faulty agents. Therefore, the algorithm satisfies the convergence criteria.
\end{proof}


It is worth noting that when $n\leq 2f+1$, it is necessary to have $X_i=X_j$ for all $i,j\in\V$ for an $f$-resilient algorithm to exist. 
This condition is also sufficient: an $f$-resilient algorithm only needs each non-faulty agent to output its input set directly.

\subsubsection{An equivalent necessary condition}
\label{sec:equiv}

In the previous section, we presented the necessary and sufficient condition on the underlying communication graph (Theorems~\ref{thm:necessity-set-intersection} and~\ref{thm:sufficiency-set-intersection}) for a Byzantine distributed set-intersection algorithm to exist, if $2f$-redundancy defined in Definition~\ref{def:set-redundancy-2} is satisfied. Let us call the condition in Theorems~\ref{thm:necessity-set-intersection} and~\ref{thm:sufficiency-set-intersection} \textit{Condition A}.
\begin{definition}[Condition A]
    A graph $G(\V,\mathcal{E})$ is said to satisfy \emph{Condition A}, if the following is true:
    \begin{enumerate}[label=\alph*)]
        \item If $\mnorm{{R}}\geq f+1$, there exists an agent $i\in L$, $\mnorm{N_i^-\cap{R}}\geq f+1$, and
        \item If $\mnorm{L}\geq f+1$, there exists an agent $i\in {R}$, $\mnorm{N_i^-\cap L}\geq f+1$.
    \end{enumerate}
    \label{def:condition-a}
\end{definition}

In this section, we present an equivalent condition to Condition A. Note that the terms \textit{agent} and \textit{node} are used interchangeably, both referring to an element in $\V$ in the communication graph.
We first introduce some definitions from \cite{vaidya2012iterative}.

\begin{definition}[Graph Decomposition \cite{vaidya2012iterative}]
    \label{def:decomposition}
    Let $H$ be a directed graph. Partition graph $H$ into non-empty strongly connected components $H_1$,...,$H_h$, where $h$ is a non-zero integer dependent on graph $H$, such that 
    \begin{enumerate}[label=\roman*)]
        \item Every pair of nodes within the same strongly connected component has directed paths in $H$ to each other, and
        \item For each pair of nodes, say $i$ and $j$, that belong to two different strongly connected components, either $i$ does not have a directed path to $j$ in $H$, or $j$ does not have a directed path to $i$ in $H$.
    \end{enumerate}
    The partition of $H$ is called a \emph{decomposition} of $H$.
    
    Construct a graph $H^d$ wherein each strongly connected component $H_k$ above is represented by vertex $c_k$, and there is an edge from vertex $c_k$ to $c_l$ if and only if the nodes in $H_k$ have directed paths in $H$ to the nodes in $H_l$.
\end{definition}

It is known that the decomposition graph $H^d$ is a directed acyclic graph \cite{dasgupta2006algorithms}.

\begin{definition}[Source Component \cite{vaidya2012iterative}]
    Let $H$ be a directed graph, and $H^d$ be its decomposition as per Definition~\ref{def:decomposition}. A strongly connected component $H_k$ of $H$ is said to be a \emph{source component} if the corresponding vertex $c_k$ in $H^d$ is not reachable from any other vertex in $H^d$.
\end{definition}

\begin{definition}[Reduced graph \cite{vaidya2012iterative}]
    \label{def:reduced-graph}
    For a given graph $G(\V,\mathcal{E})$ and $F\subset \V$, a graph $G_F(\V_F,\mathcal{E}_F)$ is a \emph{reduced graph} of $G$, if 
    \begin{enumerate}[label=\roman*)]
        \item $\V_F=\V-F$, and
        \item $\mathcal{E}_F$ is obtained by first removing all links incident on the nodes in $F$, and then removing up to $f$ other incoming links at each node in $\V_F$.
    \end{enumerate}
\end{definition}
Note that for a given graph $G$ and a given $F$, there may exist multiple reduced graphs $G_F$.

Now, we introduce the new condition, defined below as Condition B, and prove its equivalence to Condition A.
\begin{definition}[Condition B]
    A graph $G(\V,\mathcal{E})$ is said to satisfy \emph{Condition B}, if for any $F\subset\V$ such that $\mnorm{F}\leq f$, in every reduced graph $G_F$ obtained as per Definition~\ref{def:reduced-graph}, 
    the size of any source component is at least $n-2f$.
    \label{def:condition-b}
\end{definition}
We first show that Condition A implies Condition B when $\mnorm{\V}\geq 2f+2$.
\begin{lemma}
    \label{lemma:source-component-equivalent-1}
    Suppose that Condition A holds for graph $G(\V,\mathcal{E})$, and $\mnorm{\V}\geq 2f+2$. Then Condition B also holds for $G$. 
\end{lemma}
\begin{proof}
    
    The proof is by contradiction. Suppose that Condition A is satisfied, and that there exists some $F\subset\V$ with $\mnorm{F}\leq f$ and a reduced graph $G_F(\V_F,\mathcal{E}_F)$ corresponding to $F$, such that the decomposition of $G_F$ includes a source component of size less than $n-2f$.
    Suppose the source component in $G_F$ is $L$, with $\mnorm{L}<n-2f$. Let $R=\V\backslash(L\cup F)$. $L$, $R$, and $F$ forms a partition of $G$. Since $\mnorm{F}\leq f$ and $\mnorm{L}<n-2f$, $\mnorm{R}\geq f+1$. Since $L$ is a source component in $G_F$, it follows that there are no direct links in $\mathcal{E}_F$ from any node in $R$ to $L$. By the definition of a reduced graph, there are no more than $f$ links in $\mathcal{E}$ from the nodes in $R$ to any node in $L$ in graph $G$. This contradicts part a) of Condition A. Hence, the proof.
\end{proof}

\begin{corollary}
    If $n>3f$, Condition A implies that there is only one source component in any reduced graph $G_F$ of graph $G$.
    \label{corollary:1}
\end{corollary}
\begin{proof}
    The proof is by contradiction. Suppose Condition A holds for graph $G(\V,\mathcal{E})$ with $n>3f$, and there is a set $F\subset\V$ with $\mnorm{F}\leq f$, such that there are two source components $L$ and $R$. By Lemma~\ref{lemma:source-component-equivalent-1} and $n>3f$, we have $\mnorm{L}\geq n-2f\geq f+1$ and similarly, $\mnorm{R}\geq f+1$. Let $C = \V\backslash(L\cup R\cup F)$. $L$, $R\cup C$, and $F$ forms a partition of $G$. 
    
    Since $\mnorm{R}\geq f+1$, $\mnorm{R\cup C}\geq f+1$. By Condition A, there exist an agent $i\in L$, such that $\mnorm{N_i^-\cap(R\cup C)}\geq f+1$. However, since $L$ is a source component in the reduced graph $G_F$, for any agent $i\in L$, there are at most $f$ incoming neighbors from agents in $R\cup C$ in $G$. In other words, $\mnorm{N_i^-\cap(R\cup C)}\leq f$. Contradiction. Hence, the proof.
\end{proof}

\begin{lemma}
    \label{lemma:source-component-equivalent-1-alt}
    Suppose that Condition A holds for graph $G(\V,\mathcal{E})$, and $\mnorm{\V}\geq 2f+2$. For any $F\subset\V$ such that $\mnorm{F}=\phi\leq f$, in every reduced graph $G_F$ obtained as per Definition~\ref{def:reduced-graph}, 
    the size of any source component must be at least $n-\phi-f$.
\end{lemma}
\begin{proof}
    The proof largely follows the proof of Lemma~\ref{lemma:source-component-equivalent-1}. 

    The proof is by contradiction. Suppose that there exists some $F\subset\V$ with $\mnorm{F}=\phi\leq f$, and a reduced graph $G_F(\V_F,\mathcal{E}_F)$ corresponding to $F$, such that the decomposition of $G_F$ includes a source component of size less than $n-\phi-f$.

    Suppose the source component in $G_F$ is $L$, with $\mnorm{L}<n-\phi-f$. Let $R=\V\backslash(L\cup F)$. $L$, $R$, and $F$ forms a partition of $G$. Since $\mnorm{F}=\phi$, $\mnorm{R}>f$, or $\mnorm{R}\geq f+1$. Since $L$ is a source component in $G_F$, it follows that there are no direct links in $\mathcal{E}_F$ from any node in $R$ to $L$. By the definition of a reduced graph, there are no more than $f$ links $\mathcal{E}$ from the nodes in $R$ to any node in $L$ in graph $G$. This contradicts part a) of Condition A. Hence, the proof.
\end{proof}

Note that Corollary~\ref{corollary:1} can be proved with Lemma~\ref{lemma:source-component-equivalent-1-alt} as well.\footnote{The proof is by contradiction. Suppose Condition A holds for graph $G$ with $n>3f$, and there exists a set $F\subset\V$ with $\mnorm{F}\triangleq\phi\leq f$, such that there are two source components $L$ and $R$ in the reduced graph $G_F$. By Lemma~\ref{lemma:source-component-equivalent-1-alt}, $\mnorm{L}\geq n-f-\phi$ and $\mnorm{R}\geq n-f-\phi$. The three sets $L,R,F$ are disjoint, so $\mnorm{L\cup R\cup F}\geq 2n-2f-2\phi+\phi=2n-2f-\phi$. On the other hand, $\mnorm{L\cup R\cup F}\leq n$. Therefore, $2n-2f-\phi\leq n$, which implies that $n\leq 2f+\phi\leq 3f$, which contradicts our assumption of $n> 3f$.}

Now we show the other direction of the equivalency.
\begin{lemma}
    \label{lemma:source-component-equivalent-2}
    Suppose Condition B holds for graph $G(\V,\mathcal{E})$ with $\mnorm{\V}\geq 2f+2$. Condition A also holds.
\end{lemma}
\begin{proof}
    The proof is by contradiction. Suppose that Condition B holds for graph G, and part a) of Condition A is not true. That is, there exists a partition of graph $G(\V,\mathcal{E})$ into $L$, $R$, $F$ with $\mnorm{F}\leq f$, $\mnorm{R}\geq f+1$, and for any node $i\in L$, $\mnorm{N_i^-\cap R}\leq f$. 

    Consider the following transformation over $L$, $R$, and $F$: identify the sets $L_1\subset L$ and $R_1\subset R$ such that $L\backslash L_1\neq\varnothing$, $\mnorm{R\backslash R_1}\geq f+1$, and $\mnorm{L_1\cup R_1\cup F}=f$, and let $L'=L\backslash L_1$, $R'=R\backslash R_1$, and $F'=F\cup L_1\cup R_1$. $L'$, $R'$, $F'$ is also a partition of $G$, with $\varnothing\neq L'\subseteq L$, $R'\subseteq R$ and $F\subseteq F'$, $|R'|\geq f+1$, $|F|=f$, and for each $i$ in $L'$, $\mnorm{N_i^-\cap{R'}}\leq f$. Note that this transformation is always feasible because $n\geq 2f+2$. And we also have $\mnorm{L'}=n-\mnorm{F'}-\mnorm{R'}\leq n-f-(f+1)=n-2f-1<n-2f$.
    


    Since $L'$, $R'$, $F'$ form a partition of $G$, consider the reduced graph $G_{F'}$ of $G$, constructed by first removing all nodes in $F'$ from $\V$ and all edges incident to them in $\mathcal{E}$, then removing all incoming edges from nodes in $R'$ to every node in $L'$. Since for all $i\in L'$, $\mnorm{N_i^-\cap R'}\leq f$, the above steps indeed create a reduced graph of $G$ per Definition~\ref{def:reduced-graph}. Note that $L'$ has no incoming edge in $G_{F'}$ and therefore is a source component. However, its size is strictly less than $n-2f$, which contradicts Condition B. Hence, the proof.
\end{proof}

Combining Lemmas~\ref{lemma:source-component-equivalent-1} and~\ref{lemma:source-component-equivalent-2} in this section, we have shown the following:
\begin{theorem}
    Condition A and Condition B are equivalent.
\end{theorem}

\subsection{Decentralized unconstrained algorithms}
\label{sub:unconstrained}


In this section, we remove the constraints over the algorithms discussed in the previous section, that is, the algorithm can maintain more information than a local set. As such, the requirement of Validity and Convergence are also removed. Instead, suppose $F\subset\V$ is the set of Byzantine agents with $\mnorm{F}\leq f$, an $f$-resilient set intersection algorithm must satisfy the following criteria:
\begin{itemize}
    \item \textbf{Correctness}: Each agent $i\in\V\backslash F$ will decide on $X_{\V\backslash F}$ as its output in finite time.
\end{itemize}

Let us denote by $G-S$ with $S\subseteq \V$ the graph by removing all nodes in $S$ and all edges containing nodes in $S$ from $G$. We will use the definition of connectivity of a directed graph in the conditions, stated as follows:

\begin{definition}[$k$-connectivity \cite{west2001introduction}] 
    A \emph{separating set} or a \emph{vertex cut} of a directed graph $G$ is a set $S\subseteq\V$ such that the graph $G-S$ is not strongly connected. A directed graph $G$ is \emph{$k$-connected} if every separating set has at least $k$ vertices.
\end{definition}







\begin{theorem}[Necessity]
    Suppose $2f$-redundancy is satisfied, and $n\geq 2f+2$. An $f$-resilient set intersection algorithm exists only if the communication graph $G$ is $(2f+1)$-connected. 
\end{theorem}

\begin{proof}
    The proof is by contradiction. Suppose an $f$-resilient set intersection algorithm exists and $2f$-redundancy is satisfied on a graph that is not $(2f+1)$-connected. That is, there exists a subset of nodes $T\subset\V$ with $\mnorm{T}=2f$, such that $G-T$ is not strongly connected. Thus, there exist two nodes $r,l\in G-T$, such that there is no directed path in $G-T$ from $r$ to $l$. Since $n\geq 2f+2$, this is always possible. $T$ is a separating set of $G$. 
    
    Let us denote an arbitrary subset $F\subset T$ with $\mnorm{F}=f$, $R=(T\backslash F)\cup\{r\}$, and $L=\V\backslash(F\cup R)$. $L,R,F$ form a partition of $G$. Note that since $n\geq 2f+2$, $\mnorm{L}\geq 1$, and $\mnorm{R}=\mnorm{T\backslash F}+1=f+1$. Let us further denote $R_0=T\backslash F$. 

    Consider input sets such that 
    for each agent $i\in T\cup\{r\}$, $y\notin X_i$. Since $\mnorm{T}=2f$, there are exactly $2f+1$ agents whose input sets do not have $y$. Also, for each agent $i\in L$, $y\in X_i$.

    Recall that $l$ can only learn information of $r$ via agents in $T$. Note that there are at most $f$ agents in $R_0$ that have a path to $l$, and at most $f$ agents in $F$ that have a path to $l$. Suppose that in an execution, all agents in $F$ are faulty, and any faulty agent $i\in F$ will always behave as if all agents in $F\cup L\cup\{r\}$ have $y$ in their input sets.
    The agent $l$ cannot distinguish the following two scenarios:
    \begin{enumerate}[nosep, label=\roman*)]
        \item The actual situation, where agents in $F$ are faulty, all other agents are non-faulty, and $y\notin X_i$ for $i\in R$ and for each agent $i\in L$, $y\in X_i$. 
        In this case, $y$ is not in $\bigcap_{i\in\V\backslash F}X_i$, the intersection of all non-faulty input sets, and should be removed from $X_l$.
        \item Another possible situation from $l$'s perspective, where the agents in $R_0$ are faulty, the remaining agents are non-faulty, and for each non-faulty agent $i$, $y\in X_i$. 
        In this case, $y$ is in $\bigcap_{i\in\V\backslash R_0}X_i$, which is the intersection of all non-faulty input sets.
    \end{enumerate}
    Note that both scenarios are compatible with the $2f$-redundancy property: in Scenario i), there are exactly $2f+1$ agents whose input sets do not have $y$, while in Scenario ii), all the non-faulty agents have $y$ in their input sets.

    For agent $l$, if its output does not have $y$, the output would not be correct in the second scenario; if its output has $y$, the output would not be correct in the first scenario. Therefore, the output of agent $l$ cannot be guaranteed to be correct in all executions. This is a contradiction with the assumption that an $f$-resilient algorithm exists. Hence, the proof.
\end{proof}

To show sufficiency, we define a \textit{directed $(x,y)$-path} for $x,y\in\V$ to be a series of edges $(x,p_1),...,(p_k,y)$, such that $p_1,...,p_k\in\V$ ($k\geq 1$) are different nodes. Also, two $(x,y)$-paths are called \textit{internally vertex-disjoint} if $x$ and $y$ are their only common nodes. We will use the following version of Menger's Theorem:

\begin{theorem}[Menger's Theorem \cite{bondy1976graph}] 
    Let $x$ and $y$ be two vertices of a directed graph $G$, such that there is no edge in $G$ from $x$ to $y$. Then the maximum number of internally vertex-disjoint directed $(x,y)$-paths in $G$ is equal to the minimum number of vertices whose deletion destroys all directed $(x,y)$-paths in $G$.
\end{theorem}

\begin{theorem}[Sufficiency]
    Suppose $2f$-redundancy is satisfied, and $n\geq 2f+2$. An $f$-resilient set intersection algorithm exists if the communication graph $G$ is $(2f+1)$-connected.
\end{theorem}

\begin{proof}
    Consider the following algorithm:
    \begin{enumerate}[label=\textbf{Step \arabic*}]
        
        \item Each agent $i$ sends its local set $X_i$ to all its outgoing neighbors, and at least $2f+1$ copies via internally vertex-disjoint paths to all other agents in the graph. The later half of this step is possible by $(2f+1)$-connectivity and Menger's Theorem.

        For each agent $i$, receive set $Z_j$ sent from each agent $j\in\V$ other than $i$:
            \begin{itemize}
                \item If $j$ is an incoming neighbor of $i$, store the set it receives directly from $j$ as $Z_j$.
                \item Or, if $f+1$ identical copies of a set are received from agent $j$ via internally disjoint paths, store the set as $Z_j$.
                \item Otherwise, 
                that agent $j$ must be Byzantine. In this case, store $Z_j=\varnothing$.
            \end{itemize}
        \item For each agent $i$,
            initialize $O_i\leftarrow\varnothing$. For every value $y\in X_i$, if there are at most $f$ agents $k$ such that $y\notin Z_k$, add $y$ to $O_i$. Output $O_i$.
    \end{enumerate}

    Note that the algorithm is synchronous, and we assumed that the $2f$-redundancy is satisfied. We now show the algorithm is $f$-resilient.

    We first show the correctness of Step 1, that is, agent $i$ can receive and determine the correct input set of every non-faulty agent $j$. What set $Z_k$ is stored at $i$ for any Byzantine agent $k$ is irrelevant and will be handled in Step 2. For any non-faulty agent $j$, by Menger's Theorem, since the communication graph is $(2f+1)$-connected, either (i) $(j,i)\in\mathcal{E}$ and $Z_j=X_j$, or (ii) there are at least $2f+1$ messages from $j$ via internally vertex-disjoint paths are received by $i$, and since there are up to $f$ Byzantine agents, at most $f$ messages may be corrupted, and at least $f+1$ messages are delivered correctly, each containing an identical copy of $X_j$, which is stored as $Z_j$. Byzantine agents may behave 
    incorrectly, so $Z_k$ for faulty agent $k$ may not be its correct input set.

    After Step 1, each non-faulty agent $i$ stores in $Z_j$ the correct local set of agent $j$ if $j$ is non-faulty, and a potentially arbitrary set if agent $j$ is Byzantine faulty. Let $F$ be the set of all Byzantine agents with $\mnorm{F}\leq f$. Now we show that Step 2 outputs the correct intersection of all non-faulty agents. Consider any non-faulty agent $i$. For any possible value $y$,
    \begin{enumerate}[nosep, label=(\roman*)]
        \item If $y\notin X_i$, since $i$ is non-faulty, $y\notin\bigcap_{i\in\V\backslash F}X_i$. By Step 2, $y$ will not be added to $O_i$.
        \item If $y\in X_i$, and $y\notin\bigcap_{i\in\V\backslash F}X_i$. By $2f$-redundancy, 
        there are at least $2f+1$ agents $k$ with $y\notin X_k$. Since there are up to $f$ Byzantine agents, whose input set is unknown to agent $i$, there should be at least $f+1$ agents $k$ with $y\notin Z_k$. This is true because $n\geq2f+2$. By Step 2, $y$ will not be added to $O_i$.
        \item If $y\in X_i$, and $y\in\bigcap_{i\in\V\backslash F}X_i$. For any non-faulty agent $j$, $Z_j=X_j$ contains $y$. Therefore, there are at most $f$ agents $k$ with $y\notin Z_k$. By Step 2, $y$ will be added to $O_i$.
    \end{enumerate}
    Therefore, $O_i$ contains and only contains values in $\bigcap_{i\in\V\backslash F}X_i$. Hence the output is correct for all non-faulty agents, and the algorithm is $f$-resilient.
\end{proof}

\subsection{Asynchronous constrained and unconstrained algorithms}

In the previous part of this section, we discussed the necessary and sufficient conditions for a synchronous $f$-resilient set intersection algorithm to exist when $2f$-redundancy is satisfied. In this part, we present without providing detailed proofs, the necessary and sufficient conditions for an asynchronous $f$-resilient algorithm to exist with redundancy.

For asynchronous systems, it can be shown that $3f$-redundancy is necessary to solve the set intersection problem correctly:
\begin{definition}[$3f$-redundancy]
    The local sets of a set $\V$ of $n$ agents are said to satisfy $3f$-redundancy, if $\bigcap_{i\in\V}X_i\neq\varnothing$, and for any point $y\notin\bigcap_{i\in\V}X_i$, there exist at least $3f+1$ agents $i$, such that $y\notin X_i$.
\end{definition}

For constrained algorithms discussed in Section~\ref{sub:constrained}, in asynchronous systems, the following condition is necessary and sufficient:

\begin{theorem}
    Suppose $3f$-redundancy is satisfied, and $n\geq 3f+2$. A constrained asynchronous set intersection algorithm exists if for any partition of $G$ into $L, R, F$, $\mnorm{F}\leq f$, the following conditions hold:
    \begin{enumerate}[ label=\alph*)]
        \item If $\mnorm{R}\geq 2f+1$, there exists an agent in $i\in L$, $\mnorm{N_i^-\cap R}\geq 2f+1$, and
        \item If $\mnorm{L}\geq 2f+1$, there exists an agent in $i\in R$, $\mnorm{N_i^-\cap L}\geq 2f+1$.
    \end{enumerate}
\end{theorem}

For unconstrained algorithms, in asynchronous systems, \noindent The following condition is necessary and sufficient:

\begin{theorem}
    Suppose $3f$-redundancy is satisfied, and $n\geq 3f+2$. An asynchronous $f$-resilient set intersection algorithm exists if the communication graph is $(2f+1)$-connected.
\end{theorem}

\section{On general communication graphs}
\label{sec:generalized-graph}

In this section, we discuss the redundancy property needed on a generalized communication graph $G$ such that an $f$-resilient Byzantine set intersection algorithm exists. We consider a synchronous system in this section. However, these results can also be extended to asynchronous systems.

Recall that we defined $\H$ to be the set of all non-faulty agents in a given execution. $X_i$ is the local set of agent $i$. Specifically, consider the following redundancy property.

\begin{definition}[$(f, G)$-redundancy]
    Given a communication graph $G(\V,\mathcal{E})$, the sets of the agents are said to satisfy $(f,G)$-redundancy, if for any subset of agents $F\subset V$ with $\mnorm{F}\leq f$, if $y\notin\bigcap_{i\in \V\backslash{F}}X_i$, there is at least one agent $j$ in each source component in every reduced graph $G_F$, such that $y\notin X_j$. 
    \label{def:redundancy-fg}
\end{definition}
Let us denote the redundancy property defined above by \textbf{Property D}, which will be used in a later section.

Similar to Section~\ref{sec:set-inter}, we consider a class of constrained algorithms. Specifically, the algorithms are iterative, need to satisfy validity and convergence criteria, and each agent is constrained to maintain only a local set with no additional state. Now, we show that $(f,G)$-redundancy is necessary for this class of algorithms.

\begin{theorem}[Necessity]
    \label{thm:necessity-fg}
    For any given communication graph $G$,
    an $f$-resilient Byzantine set intersection algorithm exists only if $(f,G)$-redundancy is satisfied.
\end{theorem}

\begin{proof}
    The proof is by contradiction. Suppose there exists a set intersection algorithm on $G(\V,\mathcal{E})$, and $(f,G)$-redundancy is not satisfied. That is, there exists a subset $F_0\subset\V$ with $\mnorm{F_0}\leq f$, such that for some point $y\notin \bigcap_{i\in\V\backslash F_0}X_i$, there exists a source component $L_0$ in some reduced graph $G_{F_0}$, such that for every agent $i\in L_0$, $y\in X_i$. Obviously $\mnorm{L_0}\geq1$. Note that since $y\in\bigcap_{i\in L_0}X_i$ but $y\notin\bigcap_{i\in\V\backslash F_0}$, it follows that $\mnorm{R_0}\geq 1$ and $y\notin\bigcap_{i\in R_0}X_i$, where $R_0\triangleq\V\backslash(F_0\cup L_0)$.

    Note that $F_0, L_0, R_0$ form a partition of $G$. 
    Since $L_0$ is a source component in $G_{F_0}$, there are at most $f$ incoming neighbors from $R_0$ for any node in $L_0$. 

    Since $(f,G)$-redundancy is not satisfied, it is possible to construct the following input of a set intersection problem on $G$, such that for $y$, $y\in X_i$ for all $i\in L_0$, while $y\notin X_j$ for all $j\in R_0\cup F_0$.

    For any agent $l\in L_0$, there are at most $f$ agents in $R_0$ that can send information to $l$. Since $\mnorm{F_0}\leq f$, $l$ can only receive information from at most $f$ agents in $F_0$ as well. Suppose that any faulty agent in $F_0$ will always send to $l$ a set that contains $y$. Any agent $l\in L_0$ cannot distinguish the following two scenarios:    

    \begin{enumerate}[nosep, label=\alph*)]
        \item 
        The actual situation, where all agents in $F_0$ are faulty,
        all agents in $R_0$ are not faulty, and 
        $y\notin X_i$ for all $i\in R_0$. In this case, $y$ is not in $\bigcap_{i\in\V\backslash F_0}X_i$, the intersection of local sets of all non-faulty agents, and should be removed from $X_l$ as well.
        \item 
        An equivalent possible situation from $l$'s perspective, where the agents in $N_l^-\cap R_0$ are faulty, the remaining agents are non-faulty, and for each non-faulty agent $i$, $y\in X_i$. 
        In this case, $y$ is in $\bigcap_{i\in\V\backslash(N_l^-\cap R_0)}X_i$, the intersection of local sets of all non-faulty agents, and should be kept in $X_l$ as well.
    \end{enumerate}
    For any set intersection algorithm on $G$, any agent $l\in L_0$ cannot distinguish Scenarios a) and b). Only one of the two scenarios can be true. But from $l$'s point of view, the two are identical.

    Let us assume that the validity condition of the algorithm is satisfied, that is, $\forall t$, $X_l^t\subseteq X_l^{t-1}$. 
    If agent $l$ chooses to keep $y$ in $X_l$ in all iterations, its output would not satisfy the convergence condition in scenario a).
    On the other hand, if agent $l$ chooses to remove $y$ from $X_l$ at any iteration $t$, its output would not satisfy the convergence condition in scenario b).

    Thus, agent $l$'s output cannot be guaranteed to be correct in all executions. This is a contradiction with the assumption that a correct algorithm exists. Thus, $(f,G)$-redundancy is necessary.
\end{proof}


\begin{theorem}[Sufficiency]
    \label{thm:sufficiency-fg}
    For any given communication graph $G$, an $f$-resilient Byzantine set intersection algorithm exists if $(f,G)$-redundancy is satisfied.
\end{theorem}
\begin{proof}
    The proof is by construction. Consider the following algorithm:

    Each agent $i$ sets the set $X_i^0$ to be its input set $X_i$. In each iteration $t$,
    \begin{enumerate}[label=\textbf{Step \arabic*}]
        \item Each agent $i$ sends the set $X_i^t$ to all agents $j$ with $(i,j)\in\mathcal{E}$.
        \item Each agent $j$ receives sets from all agents in $N_j^-$. For each agent $k\in N_j^-$, denote the set $j$ receives from $k$ by $Z_k^t$. For a non-faulty agent $k$, $Z_k^t=X_k^t$. For a Byzantine agent $l$, $Z_l^t$ can be an arbitrary set.
        \item Each agent $i$ updates its set by the following rule: Initiate the set $Y_j^t$ to be empty. For every $y\in X_j^t$, if there are more than $f$ agents $j\in N_i^-$ such that $y\notin Z_j^t$, add $y$ to $Y_j^t$. Then, the local set of $i$ is updated as
        \begin{align}
            X_i^{t+1}\leftarrow X_i^t\backslash Y_i^t.
            \label{eqn:update-general-sufficiency}
        \end{align}

    \end{enumerate}

    Now we show that the above algorithm is indeed $f$-resilient. The validity is already guaranteed by the update \eqref{eqn:update-general-sufficiency}. Also, by validity, for any value $y\notin X_j$ for any non-faulty agent $j$, $y$ will not be added back to any $X_j^t$ during the process of the algorithm. Therefore, for convergence, we only need to consider the values that are in the local set of any non-faulty agent.
    
    For convergence, consider an arbitrary non-faulty agent $j$, and an arbitrary value $y\in X_j^t$ at an arbitrary iteration $t$, in an arbitrary execution where the set of faulty agents is $F\subset\V$ with $\mnorm{F}\triangleq\phi\leq f$. There are two possible cases:
    \begin{enumerate}
        \item[Case 1] $y\in\bigcap_{i\in\V\backslash F}X_i$, therefore $y$ is in the input set $X_j$ and should be kept in the local set of every non-faulty agent $j$. Therefore, only agents in $F$ can send sets without $y$. For agent $j$, since $\mnorm{F}\leq f$, there are at most $f$ agents that may send sets without $y$, and $y$ will not be removed from the set $X_j^t$ in any iteration $t$.
        \item[Case 2] $y\notin\bigcap_{i\in\V\backslash F}X_i$, and therefore $y$ should be eventually removed from the local sets of all non-faulty agents. 
        To make the correctness argument in Case 2, let us define the following indicator function $\mathds{1}_i^j(y)$ with respect to $j$ and its incoming neighbor $i\in N_j^-$ 
        \begin{align}
            \mathds{1}_i^j(y) = \begin{cases}
                1, \qquad \textrm{if $y\in Z_i^t$}, \\
                0, \qquad \textrm{if $y\notin Z_i^t$}.
            \end{cases}
            \label{eqn:1-function}
        \end{align}
        Consider sorting the values $\mathds{1}_i^j(y)$ for all $i\in N_j^-$, and remove the largest and smallest $f$ values from the sorted list (breaking ties arbitrarily). Let us denote the set of agents sending the remaining values by $S_j$. Then it is obvious that we have the following claim:
        \begin{claim}
            There are more than $f$ agents $k\in N_j^-$ such that $y\notin Z_k$, if and only if the $\mathds{1}_i^j(y)$'s for $i\in S_j$ contain at least one 0.
            \label{claim:1}
        \end{claim}

        Consider the sorted list of values $\mathds{1}_i^j(y)$'s for $y$. By removing the largest and smallest $f$ values from the list, either (i) all faulty values are removed, which is possible since $\mnorm{F}\leq f$, or (ii) at least 1 value from faulty agents is kept. There are in total $2f$ values removed

        In scenario (i), all remaining values are from non-faulty agents. In scenario (ii), there are two cases:
        \begin{enumerate}
            \item If the remaining values contain both 0 and 1, it implies there are $f$ 0's and $f$ 1's that are removed, and each remaining faulty value can be replaced by a removed but equal value that is from a non-faulty agent. 
            \item If the remaining values contain only 0 or 1, it implies there are $f$ 0's or $f$ 1's that are removed, respectively, and each remaining faulty value can also be replaced by a removed but equal value that is from a non-faulty agent. 
        \end{enumerate}
        In either case, the remaining $\mnorm{N_j^-}-2f$ values can be viewed as if each of them is from a distinct non-faulty agent in $N_j^-$, and there are the same number of 0's and 1's as in from $S_j$. Denote the set of these agents by $S'_j$.
        
        Furthermore, let $\phi\triangleq\mnorm{F}$, and we can add $f-\phi$ to the values to $S'_j$ in the following way:
        \begin{enumerate}
            \item If there are both 0's and 1's in the values from $S_j'$, at most $\phi$ of the removed values were added back to $S_j'$ in the previous step, and at least $f-\phi$ non-faulty 1's still remain removed. 
            We can add $f-\phi$ additional 1's from the removed non-faulty values, so that there are $\mnorm{N_j^-}-f-\phi$ non-faulty values and they contain both 0 and 1.
            \item If there are only 0's or 1's in the values from $S_j'$, at most $\phi$ of the removed values were added back to $S_j'$ in the previous step, and at least $f-\phi$ non-faulty 0's or 1's still remain removed. 
            We can add $f-\phi$ additional 0's or 1's, respectively, from the removed non-faulty values, so that there are $\mnorm{N_j^-}-f-\phi$ non-faulty values and they contain only 0 or 1, respectively.
        \end{enumerate}
        Denote the set of these $\mnorm{N_j^-}-f-\phi$ non-faulty agents by $S_j''$. 
        The above transformations from $S_j$ to $S_j'$ to $S_j''$ give us the following claim:
        \begin{claim}
            \label{claim:2}
            $\mathds{1}_i^j(y)$'s for all $i\in S_j''$ contains at least one 0, if and only if $\mathds{1}_i^j(y)$'s for all $i\in S_j$ contains at least one 0.
        \end{claim}
        Recall the rule of deciding whether to remove $y$ in Step 3 of the algorithm. By Claims~\ref{claim:1} and~\ref{claim:2}, $y$ should be removed in the current iteration, if and only if $\mathds{1}_i^j(y)$'s for all $i\in S_j''$ contains at least one 0.

        By the definition of a reduced graph, it follows that there is a corresponding reduced graph $G_F$ in which agent $j$ would receive and only receive values from its incoming neighbors, and those agents are the same as non-faulty agents in $S_j''$.\footnote{Note that if some agent $j$ has $\mnorm{N_j^t}-f-\phi\leq0$, by the definition of a reduced graph, there exists a reduced graph $G_F$ such that $j$ itself is a source component. By $(f,G)$-redundancy, $y$ is not in $X_j$, and therefore the algorithm does not need to consider removing $y$ from the local set of $j$, skipping this discussion.}
        There are two cases: 
        \begin{enumerate}[label=Case (\roman*)]
            \item $j$ is in a source component in the corresponding reduced graph $G_F$. By $(f,G)$-redundancy, there is at least one agent $i$ with $y\notin X_i^0$ in every source component of $G_{F}$ for any $F\subset\V$ with $\mnorm{F}\leq f$. For every outgoing neighbor $k$ of $i$ in the same source component, $i\in N_k^-$, and therefore, values from $S_k''$ of agent $k$ contain at least one 0. In the current iteration $t$, if $y\in X_k^t$, $y$ will be removed from agent $k$'s set.
            \item $j$ is not in a source component in the corresponding reduced graph $G_F$. Following the discussions of Case (i), eventually all agents that belong to a source component in $G_{F}$ will remove $y$ from their local sets in finite iterations. Since $j$ is not in a source component, there exists a path in $G_F$ from an agent $k$ in a source component to $j$. At least one more agent on that path will remove $y$ from its local set in each iteration after agent $k$ removes $y$ from its local set. Therefore, $j$ will remove $y$ from its local set in finite iterations as well.
        \end{enumerate}
        Combining the two cases above, agent $j$, and therefore all non-faulty agents will remove $y$ from their local set within finite iterations.
    \end{enumerate}
    Therefore, for all non-faulty agents, all values in $\bigcap_{i\in\V\backslash F}X_i$ will be kept, and all values not in $\bigcap_{i\in\V\backslash F}X_i$ will be removed within finite iterations. In other words, the proposed algorithm is $f$-resilient.
\end{proof}

\section{From set intersection to optimization}
\label{sec:to-optimization}

Inspired by the relationship between the set intersection and optimization problems discussed in Section~\ref{sec:optimization}, results in previous sections can also provide us insights into the necessary and sufficient conditions for solving a Byzantine distributed optimization problem in a generalized network.

We consider a synchronous system in this section. However, these results can also be extended to asynchronous systems.
 
\subsection{With $2f$-redundancy}

We first examine the decentralized case with $2f$-redundancy. Consider the following assumptions:

\begin{assumption}
    \label{assum:convex-diff}
    The cost functions of non-faulty agents are convex and differentiable.
\end{assumption}

\begin{assumption}
    \label{assum:singleton-minimum}
    The minimum set of the aggregate of all cost functions $\arg\min_x\sum_{i\in\V}Q_i(x)$ only has 1 point.
\end{assumption}

Let us recapture some related results bridging optimization and set intersection in the previous sections. We define $f$-resilience and $2f$-redundancy (Property A) for Byzantine optimization problems in Definitions~\ref{def:f-resilience} and \ref{def:redundancy-exact}, respectively. We also define $f$-resilience and $2f$-redundancy (Properties B and C) for Byzantine set intersection problems in Definitions~\ref{def:f-resilient-set-intersection}, ~\ref{def:set-redundancy-1} and \ref{def:set-redundancy-2}, respectively. We then show in Theorem~\ref{thm:necessity-set-intersection} the necessity condition (Condition A in Definition~\ref{def:condition-a}) for an $f$-resilient set intersection algorithm to exist if $2f$-redundancy is satisfied. We also show in Lemma~\ref{lemma:alternative-condition-necessity} that with Assumption~\ref{assum:convex-diff}, Property B on minimum sets of the cost functions, which is equivalent to Property A on cost functions, derives Property C on the minimum sets of the cost functions, or in short, Property A on cost functions derives Property C on minimum sets of the cost functions.

For Byzantine optimization problems on a decentralized distributed system with $2f$-redundancy, we have the following results:

\begin{theorem}[Necessity, optimization]
    In a decentralized distributed system for Byzantine optimization, 
    suppose $2f$-redundancy (Property A, Definition~\ref{def:redundancy-exact}) is satisfied. A deterministic $f$-resilient algorithm per Definition~\ref{def:f-resilience} exists for all cost functions satisfying Assumptions~\ref{assum:convex-diff} and~\ref{assum:singleton-minimum}, only if Condition A per Definition~\ref{def:condition-a} is satisfied.
\end{theorem}

\begin{proof}
    The proof is by contradiction. Suppose $2f$-redundancy in Property A is satisfied,  Condition A is not satisfied, and a deterministic $f$-resilient algorithm exists. 
    
    With Assumption~\ref{assum:convex-diff}, the cost functions of non-faulty agents are convex and differentiable. By Lemma~\ref{lemma:alternative-condition-necessity}, Property A indicates Property C per Definition~\ref{def:set-redundancy-2} is also satisfied for the minimum sets $X_i$'s of the cost functions of each agent $i$. Property A with Assumption~\ref{assum:singleton-minimum} also indicates that there exists $x^*$, such that
    \begin{align}
        \{x^*\} = \arg\min_x\sum_{i\in\V}Q_i(x) = \arg\min_x\sum_{i\in S}Q_i(x)
    \end{align}
    for all $S\subseteq\V$ with $\mnorm{S}\geq n-2f$. This implies $\bigcap_{i\in S}\arg\min_xQ_i(x)\neq\varnothing$ for all $S$ with $\mnorm{S}\geq n-2f$ as well. By Lemma~\ref{lemma:minimum-set-intersection}, we also have
    \begin{align}
        \bigcap_{i\in S}\arg\min_xQ_i(x) = \arg\min_x\sum_{i\in S}Q_i(x)
    \end{align}
    for all $S$.
    
    A deterministic $f$-resilient optimization algorithm exists, meaning it will output $x^*$ despite the presence of up to $f$ Byzantine agents. It follows that this algorithm is also $f$-resilient per Definition~\ref{def:f-resilient-set-intersection} for set intersection problems over sets $X_i$'s. 

    Following the same argument in the proof of Theorem~\ref{thm:necessity-set-intersection}, Condition A can also be shown to be necessary for an $f$-resilient algorithm to exist where the intersection of the sets of all non-faulty agents is known to be a singleton.
    However, as argued above, the $f$-resilient set intersection algorithm exists without Condition A, which is a contradiction. Hence, the proof.
\end{proof}

\begin{theorem}[Sufficiency, optimization]
    In a decentralized distributed system for Byzantine optimization, suppose Assumption~\ref{assum:convex-diff} is satisfied. Suppose $2f$-redundancy (Property A, Definition~\ref{def:redundancy-exact}) is satisfied. A deterministic $f$-resilient algorithm per Definition~\ref{def:f-resilience} exists if Condition A per Definition~\ref{def:condition-a} is satisfied.
\end{theorem}

\begin{proof}
    The proof is by reduction. With Assumption~\ref{assum:convex-diff}, $2f$-redundancy in Property A 
    indicates that $\bigcap_{i\in\V}\arg\min_xQ_i(x)\neq\varnothing$. By Lemma~\ref{lemma:minimum-set-intersection}, we also have 
    \begin{align}
        \bigcap_{i\in S}\arg\min_xQ_i(x) = \arg\min_x\sum_{i\in S}Q_i(x)
    \end{align}
    for all $S\subseteq\V$. Therefore, the output of an $f$-resilient optimization algorithm is a point in $\bigcap_{i\in\V\backslash F}\arg\min_xQ_i(x)$, where $F$ is the set of all faulty agents in an execution.

    With Assumption~\ref{assum:convex-diff}, by Lemma~\ref{lemma:alternative-condition-necessity}, $2f$-redundancy indicates Property C per Definition~\ref{def:set-redundancy-2}. With Property C and Condition A, by Theorem~\ref{thm:sufficiency-set-intersection}, there exists a deterministic $f$-resilient set intersection algorithm. Therefore, consider the following algorithm:
    \begin{enumerate}[label=\textbf{Step \arabic*}]
        \item Each agent $i$ computes the minimum set $X_i$ of its local cost function $Q_i(x)$.
        \item The system runs an $f$-resilient set intersection algorithm to find the set intersection $\bigcap_{i\in\V\backslash F}X_i$ among all non-faulty agents.
        \item Each agent $i$ output the geometric center of $\bigcap_{i\in\V\backslash F}X_i$.
    \end{enumerate}
    By designating the output to be the geometric center, each non-faulty agent agrees on the same point in $\bigcap_{i\in\V\backslash F}\arg\min_xQ_i(x)$. Therefore, the proposed algorithm is $f$-resilient for Byzantine optimization.
\end{proof}

\subsection{With a general communication graph}

Let us first define the corresponding redundancy condition to Definition~\ref{def:redundancy-fg} in optimization problems.

\begin{definition}[$(f,G)$-redundancy]
    \label{def:redundancy-fg-opt}
    Given a communication graph $G(\V,\mathcal{E})$ of a distributed system for optimization, the cost functions of the agents are said to satisfy $(f,G)$-redundancy, if for any subset of agent $F\subset\V$ with $\mnorm{F}\leq f$, for every source component $S$ in every reduced graph $G_F$, we have
    \begin{align}
        \arg\min_x\sum_{i\in S}Q_i(x) = \arg\min_x\sum_{i\in\V\backslash F}Q_i(x).
        \label{eqn:redundancy-fg-opt}
    \end{align}
\end{definition}
Recall that we denote the $(f,G)$-redundancy in Definition~\ref{def:redundancy-fg} for set intersection by {Property D} in Section~\ref{sec:generalized-graph}. Let us also denote the $(f,G)$-redundancy in Definition~\ref{def:redundancy-fg-opt} by \textbf{Property E}. 

Consider the following assumption:
\begin{assumption}
    \label{assum:minimum-set-intersection}
    The intersection of the minimum sets of the cost functions of all agents is not empty. That is, 
    \begin{align}
        \bigcap_{i\in\V}\arg\min_xQ_i(x)\neq\varnothing.
    \end{align}
\end{assumption}

Recall that we defined 
\begin{align}
    X_i &\triangleq \arg\min_xQ_i(x), \\
    X_S &\triangleq \arg\min_x\sum_{i\in S}Q_i(x),
\end{align}
for all agent $i\in\V$ and subset $S\subseteq\V$. 
We have
\begin{lemma}
    Suppose Assumption~\ref{assum:minimum-set-intersection} is satisfied. On any given communication graph $G$, for a group of $n$ cost functions, Property E is satisfied if and only if Property D is satisfied for the minimum set of each cost function.
    \label{lemma:redundancy-fg-equal}
\end{lemma}
\begin{proof}
    Suppose Property E is satisfied. By Lemma~\ref{lemma:minimum-set-intersection}, with Assumption~\ref{assum:minimum-set-intersection}, for any subset $S\subseteq\V$ we have
    \begin{align}
        \arg\min_x\sum_{i\in S}Q_i(x) = \bigcap_{i\in S}\arg\min_xQ_i(x).
    \end{align}
    Therefore, from \eqref{eqn:redundancy-fg-opt} we have
    \begin{align}
        \bigcap_{i\in S}\arg\min_xQ_i(x) = \bigcap_{i\in\V\backslash F}\arg\min_xQ_i(x),
    \end{align}
    or in short, $\bigcap_{i\in S}X_i = \bigcap_{i\in\V\backslash F}X_i$.
    That is, for any $y\notin\bigcap_{i\in\V\backslash F}X_i$, for every source component $S$ in $G_F$, there is at least one agent $i\in S$, $y\notin X_i$. In other words, Property D is also satisfied for the minimum set of each cost function.

    Now consider the other direction. Suppose Property D is satisfied for the minimum set of each cost function. That is, for any $F\subset\V$ with $\mnorm{F}\geq f$, if $y\notin\bigcap_{i\in\V\backslash F}X_i$, there is at least one agent $j$ in any source component $S$ in every reduced graph $G_F$, such that $y\notin X_j$. Therefore, we have $y\notin\bigcap_{i\in S}X_i$. Since this is true for all $y\notin\bigcap_{i\in\V\backslash F}X_i$, we have $\bigcap_{i\in S}X_i \subseteq \bigcap_{i\in\V\backslash F}X_i$. On the other hand, since $S\subseteq \V\backslash F$, we also have $\bigcap_{i\in S}X_i \supseteq \bigcap_{i\in\V\backslash F}X_i$. Therefore,  $\bigcap_{i\in S}X_i = \bigcap_{i\in\V\backslash F}X_i$. 
    By Lemma~\ref{lemma:minimum-set-intersection}, with Assumption~\ref{assum:minimum-set-intersection}, for any subset $S\subseteq\V$ we have
    \begin{align}
        \arg\min_x\sum_{i\in S}Q_i(x) = \bigcap_{i\in S}\arg\min_xQ_i(x).
    \end{align}
    Therefore, 
    \begin{align}
        \arg\min_x\sum_{i\in S}Q_i(x) = \arg\min_x\sum_{i\in\V\backslash F}Q_i(x).
    \end{align}
    That is, Property E is also satisfied.
\end{proof}

\begin{theorem}[Necessity]
    For any given communication graph $G$ with $n\geq f+1$, an $f$-resilient Byzantine optimization algorithm exists for all cost functions satisfying Assumptions~\ref{assum:singleton-minimum} and \ref{assum:minimum-set-intersection}, only if Property E, $(f,G)$-redundancy for optimization, is satisfied.
\end{theorem}
\begin{proof}
    The proof is by contradiction. On any given communication graph $G$, suppose Property E is not satisfied, and an $f$-resilient optimization algorithm exists. By Lemma~\ref{lemma:redundancy-fg-equal}, with Assumption~\ref{assum:minimum-set-intersection}, Property D is also not satisfied for the minimum set of each cost function.

    Since an $f$-resilient optimization algorithm exists, the algorithm is guaranteed to output a point in the minimum of the aggregate cost functions of all non-faulty agents in any execution, so long as the number of faulty agents is no more than $f$. By Assumption~\ref{assum:singleton-minimum}, the same algorithm is guaranteed to find the intersection of the minimum point of all non-faulty agents. In other words, an $f$-resilient Byzantine set intersection algorithm exists on $G$ so long as the intersection of non-faulty agents is a singleton, even if Property D is not satisfied.

    Following the same argument in the proof of Theorem~\ref{thm:necessity-fg}, Property D is also necessary for an $f$-resilient algorithm to exist where the intersection of the sets of all non-faulty agents is known to be a singleton. However, the $f$-resilient set intersection algorithm exists without Property D, which is a contradiction. Hence, the proof.
\end{proof}

\begin{theorem}[Sufficiency]
    Suppose Assumption~\ref{assum:minimum-set-intersection} is satisfied. 
    For any given communication graph $G$, an $f$-resilient Byzantine optimization algorithm exists if Property E is satisfied.
\end{theorem}
\begin{proof}
    The proof is by reduction. Since Property E is satisfied, with Assumption~\ref{assum:minimum-set-intersection} and Lemma~\ref{lemma:minimum-set-intersection}, the output of an $f$-resilient optimization algorithm is a point in $\bigcap_{i\in V\backslash F}\arg\min_xQ_i(x)$, where $F$ is the set of all faulty agents in an execution.
    
    Since Property E is satisfied, with Assumption~\ref{assum:minimum-set-intersection}, Property D is also satisfied for the minimum sets of the cost functions. By Theorem~\ref{thm:sufficiency-fg}, there exists a deterministic $f$-resilient set intersection algorithm. Therefore, consider the following algorithm:
    \begin{enumerate}[label=\textbf{Step \arabic*}]
        \item Each agent $i$ computes the minimum set $X_i$ of its local cost function $Q_i(x)$.
        \item The system runs an $f$-resilient set intersection algorithm to find the set intersection $\bigcap_{i\in\V\backslash F}X_i$ among all non-faulty agents.
        \item Each agent $i$ output the geometric center of $\bigcap_{i\in\V\backslash F}X_i$.
    \end{enumerate}
    By designating the output to be the geometric center, each non-faulty agent agrees on the same point in $\bigcap_{i\in\V\backslash F}\arg\min_xQ_i(x)$. Therefore, the proposed algorithm is $f$-resilient for Byzantine optimization.
\end{proof}
\section{Summary}

In this report, we discussed a Byzantine distributed set intersection problem. We demonstrated the relationship between this problem and the Byzantine distributed optimization problem, and propose a $2f$-redundancy property following $2f$-redundancy in Byzantine optimization. We provide the necessary and sufficient condition on the communication graph of a distributed system if $2f$-redundancy is satisfied, and then provide an equivalent condition. We then extend the problem to arbitrary communication graphs and provide a necessary condition on any given graph such that an algorithm exists. Lastly, we present the necessary and sufficient conditions for Byzantine optimization problems inspired by previous discussions on Byzantine set intersection problems. The results we provide are for synchronous and asynchronous systems both.

It is worth noting that the formulation of set intersection problems in our paper can also be viewed as a special type of consensus. Determining if each value $y$ is in the intersection of sets of all non-faulty agents is equivalent to reaching a consensus on the product of the values $\mathds{1}_j(y)$'s of each non-faulty agent $j$, where the function is defined as follows:
\begin{align}
    \mathds{1}_j(y) = \begin{cases}
        1, \qquad \textrm{if $y\in X_j$}, \\
        0, \qquad \textrm{if $y\notin X_j$},
    \end{cases}
    \label{eqn:1-function-x}
\end{align}
similar to what we have in \eqref{eqn:1-function}. It is also worth noting that determining if $y$ is in the union of the sets of all non-faulty agents is equivalent to reaching a consensus on the product of $1-\mathds{1}_j(y)$ of each non-faulty agent $j$.




\bibliographystyle{plainnat}
\bibliography{bib}

\end{document}